\documentclass[a4paper]{elsarticle}

\usepackage{mathtools}
\usepackage{amssymb}
\usepackage{hyperref}

\usepackage{booktabs}

\newcommand*{\N}{\mathbb{N}}

\newcommand*{\A}{\mathcal{A}}
\newcommand*{\B}{\mathcal{B}}
\newcommand*{\C}{\mathfrak{C}}
\newcommand*{\D}{\mathfrak{D}}

\newcommand*{\ra}{\rightarrow}

\newcommand*{\Xtt}{X_{\emptyword}^2}
\newcommand*{\Ytt}{Y_{\emptyword}^2}

\newcommand*{\rat}{\mathcal{R}}
\newcommand*{\CF}{\mathcal{CF}}
\newcommand*{\etol}{\mathcal{ET}0\mathcal{L}}
\newcommand*{\edtol}{\mathcal{EDT}0\mathcal{L}}

\newcommand*{\OC}{\mathcal{O}}
\newcommand*{\lin}{\mathcal{LIN}}
\newcommand*{\ind}{\mathcal{I}}
\newcommand*{\csl}{\mathcal{CSL}}

\newcommand*{\WP}{\operatorname{WP}}

\newtheorem{theorem}{Theorem}
\newtheorem{proposition}[theorem]{Proposition}
\newtheorem{corollary}[theorem]{Corollary}
\newdefinition{example}[theorem]{Example}
\newdefinition{conjecture}[theorem]{Conjecture}
\newproof{proof}{Proof}

\usepackage{tikz}
\usetikzlibrary{arrows.meta}
\usetikzlibrary{calc}
\usetikzlibrary{decorations.pathmorphing}

\usetikzlibrary{intersections}

\makeatletter

\pgfkeys{/pgf/.cd,
  rectangle corner radius/.initial=3pt
}
\newif\ifpgf@rectanglewrc@donecorner@
\def\pgf@rectanglewithroundedcorners@docorner#1#2#3#4{%
  \edef\pgf@marshal{%
    \noexpand\pgfintersectionofpaths
      {%
        \noexpand\pgfpathmoveto{\noexpand\pgfpoint{\the\pgf@xa}{\the\pgf@ya}}%
        \noexpand\pgfpathlineto{\noexpand\pgfpoint{\the\pgf@x}{\the\pgf@y}}%
      }%
      {%
        \noexpand\pgfpathmoveto{\noexpand\pgfpointadd
          {\noexpand\pgfpoint{\the\pgf@xc}{\the\pgf@yc}}%
          {\noexpand\pgfpoint{#1}{#2}}}%
        \noexpand\pgfpatharc{#3}{#4}{\cornerradius}%
      }%
    }%
  \pgf@process{\pgf@marshal\pgfpointintersectionsolution{1}}%
  \pgf@process{\pgftransforminvert\pgfpointtransformed{}}%
  \pgf@rectanglewrc@donecorner@true
}
\pgfdeclareshape{rectangle with rounded corners}
{
  \inheritsavedanchors[from=rectangle] 
  \inheritanchor[from=rectangle]{north}
  \inheritanchor[from=rectangle]{north west}
  \inheritanchor[from=rectangle]{north east}
  \inheritanchor[from=rectangle]{center}
  \inheritanchor[from=rectangle]{west}
  \inheritanchor[from=rectangle]{east}
  \inheritanchor[from=rectangle]{mid}
  \inheritanchor[from=rectangle]{mid west}
  \inheritanchor[from=rectangle]{mid east}
  \inheritanchor[from=rectangle]{base}
  \inheritanchor[from=rectangle]{base west}
  \inheritanchor[from=rectangle]{base east}
  \inheritanchor[from=rectangle]{south}
  \inheritanchor[from=rectangle]{south west}
  \inheritanchor[from=rectangle]{south east}

  \savedmacro\cornerradius{%
    \edef\cornerradius{\pgfkeysvalueof{/pgf/rectangle corner radius}}%
  }

  \backgroundpath{%
    \pgfkeys{
      /pgf/tips=false,
    }
    \northeast\advance\pgf@y-\cornerradius\relax
    \pgfpathmoveto{}%
    \pgfpatharc{0}{90}{\cornerradius}%
    \northeast\pgf@ya=\pgf@y\southwest\advance\pgf@x\cornerradius\relax\pgf@y=\pgf@ya
    \pgfpathlineto{}%
    \pgfpatharc{90}{180}{\cornerradius}%
    \southwest\advance\pgf@y\cornerradius\relax
    \pgfpathlineto{}%
    \pgfpatharc{180}{270}{\cornerradius}%
    \northeast\pgf@xa=\pgf@x\advance\pgf@xa-\cornerradius\southwest\pgf@x=\pgf@xa
    \pgfpathlineto{}%
    \pgfpatharc{270}{360}{\cornerradius}%
    \northeast\advance\pgf@y-\cornerradius\relax
    \pgfpathlineto{}%
  }

  \anchor{before north east}{\northeast\advance\pgf@y-\cornerradius}
  \anchor{after north east}{\northeast\advance\pgf@x-\cornerradius}
  \anchor{before north west}{\southwest\pgf@xa=\pgf@x\advance\pgf@xa\cornerradius
    \northeast\pgf@x=\pgf@xa}
  \anchor{after north west}{\northeast\pgf@ya=\pgf@y\advance\pgf@ya-\cornerradius
    \southwest\pgf@y=\pgf@ya}
  \anchor{before south west}{\southwest\advance\pgf@y\cornerradius}
  \anchor{after south west}{\southwest\advance\pgf@x\cornerradius}
  \anchor{before south east}{\northeast\pgf@xa=\pgf@x\advance\pgf@xa-\cornerradius
    \southwest\pgf@x=\pgf@xa}
  \anchor{after south east}{\southwest\pgf@ya=\pgf@y\advance\pgf@ya\cornerradius
    \northeast\pgf@y=\pgf@ya}

  \anchorborder{%
    \pgf@xb=\pgf@x
    \pgf@yb=\pgf@y%
    \southwest%
    \pgf@xa=\pgf@x
    \pgf@ya=\pgf@y%
    \northeast%
    \advance\pgf@x by-\pgf@xa%
    \advance\pgf@y by-\pgf@ya%
    \pgf@xc=.5\pgf@x
    \pgf@yc=.5\pgf@y%
    \advance\pgf@xa by\pgf@xc
    \advance\pgf@ya by\pgf@yc%
    \edef\pgf@marshal{%
      \noexpand\pgfpointborderrectangle
      {\noexpand\pgfqpoint{\the\pgf@xb}{\the\pgf@yb}}
      {\noexpand\pgfqpoint{\the\pgf@xc}{\the\pgf@yc}}%
    }%
    \pgf@process{\pgf@marshal}%
    \advance\pgf@x by\pgf@xa%
    \advance\pgf@y by\pgf@ya%
    \pgfextract@process\borderpoint{}%
    \pgf@rectanglewrc@donecorner@false
    \southwest\pgf@xc=\pgf@x\pgf@yc=\pgf@y
    \advance\pgf@xc\cornerradius\relax\advance\pgf@yc\cornerradius\relax
    \borderpoint
    \ifdim\pgf@x<\pgf@xc\relax\ifdim\pgf@y<\pgf@yc\relax
      \pgf@rectanglewithroundedcorners@docorner{-\cornerradius}{0pt}{180}{270}%
    \fi\fi
    \ifpgf@rectanglewrc@donecorner@\else
      \southwest\pgf@yc=\pgf@y\relax\northeast\pgf@xc=\pgf@x\relax
      \advance\pgf@xc-\cornerradius\relax\advance\pgf@yc\cornerradius\relax
      \borderpoint
      \ifdim\pgf@x>\pgf@xc\relax\ifdim\pgf@y<\pgf@yc\relax
       \pgf@rectanglewithroundedcorners@docorner{0pt}{-\cornerradius}{270}{360}%
      \fi\fi
    \fi
    \ifpgf@rectanglewrc@donecorner@\else
      \northeast\pgf@xc=\pgf@x\relax\pgf@yc=\pgf@y\relax
      \advance\pgf@xc-\cornerradius\relax\advance\pgf@yc-\cornerradius\relax
      \borderpoint
      \ifdim\pgf@x>\pgf@xc\relax\ifdim\pgf@y>\pgf@yc\relax
       \pgf@rectanglewithroundedcorners@docorner{\cornerradius}{0pt}{0}{90}%
      \fi\fi
    \fi
    \ifpgf@rectanglewrc@donecorner@\else
      \northeast\pgf@yc=\pgf@y\relax\southwest\pgf@xc=\pgf@x\relax
      \advance\pgf@xc\cornerradius\relax\advance\pgf@yc-\cornerradius\relax
      \borderpoint
      \ifdim\pgf@x<\pgf@xc\relax\ifdim\pgf@y>\pgf@yc\relax
       \pgf@rectanglewithroundedcorners@docorner{0pt}{\cornerradius}{90}{180}%
      \fi\fi
    \fi
  }
}

\makeatother

\makeatletter

\@ifpackagelater{tikz}{2013/12/01}{
  
}{
  
}

\makeatother%

\DeclarePairedDelimiter{\set}{\{}{\}}
\DeclarePairedDelimiterX{\gset}[2]{\{}{\}}{\,#1:#2\,}

\newcommand*{\biggg}{\bBigg@{4}}

\newcommand*{\Biggg}{\bBigg@{5}}

\makeatletter
\newcommand*{\sizeddelimiter}[2]{\bBigg@{#1}#2}
\makeatother

\newcommand*{\emptyword}{\varepsilon}
\newcommand*{\rev}{\mathrm{rev}}

\DeclarePairedDelimiterX{\pres}[2]{\langle}{\rangle}{#1\,\delimsize\vert\,\mathopen{}#2}

\tikzset{
  contained/.style={
    Hooks[{length=2pt,left}]-Computer Modern Rightarrow[{width=6pt,length=6pt}],
  },
  rcontained/.style={
    Hooks[{length=2pt,right}]-Computer Modern Rightarrow[{width=6pt,length=6pt}],
  },
  maybecontained/.style={
    ->,
  },
  incomparable/.style={
    decorate,
    decoration={
      zigzag,
      amplitude=.6mm,
      segment length=1mm,
      pre length=1mm,
      post length=1mm}
  },
}

\begin{document}

\title{A Language Hierarchy of Binary Relations}

\author{Tara Brough\corref{cor1}\fnref{fn1,fn3}}
\ead{t.brough@fct.unl.pt}

\author{Alan Cain\fnref{fn2,fn3}}
\ead{a.cain@fct.unl.pt}

\address{
Centro de Matem\'{a}tica e Aplica\c{c}\~{o}es,
Faculdade de Ci\^{e}ncias e Tecnologia,
Universidade Nova de Lisboa,
2829--516 Caparica,
Portugal}

\cortext[cor1]{Corresponding author}
\fntext[fn1]
{The first author was supported by an {\scshape FCT} post-doctoral fellowship ({\scshape SFRH}/{\sc
      BPD}/121469/2016).}

\fntext[fn2]{The second author was supported by an Investigador {\scshape FCT}
    fellowship ({\scshape IF}/01622/2013/{\scshape CP}1161/{\scshape CT}0001).}

\fntext[fn3]{Both authors were partially supported by by the Funda\c{c}\~{a}o para a Ci\^{e}ncia e a Tecnologia (Portuguese Foundation for Science and Technology) through the
projects {\scshape UID}/{\scshape MAT}/00297/2013 (Centro de Matem\'{a}tica e Aplica\c{c}\~{o}es),
{\scshape PTDC}/{\scshape MAT-PUR}/31174/2017 and
{\scshape PTDC}/{\scshape MHC-FIL}/2583/2014.}

\begin{abstract}
  Motivated by the study of word problems of monoids, we explore two ways of viewing binary relations on $X^*$ as languages.  We exhibit a hierarchy of classes of binary relations on $X^*$, according to the class of languages the relation belongs to and the chosen viewpoint.  We give examples of word problems of monoids distinguishing the various classes.  Aside from the algebraic interest, these 
 examples demonstrate that the hierarchy still holds when restricted to equivalence relations.
\end{abstract}


\maketitle

\section{Introduction}

In several applications of language theory in algebra and combinatorics, the issue arises of representing a relation in
a way that is recognizable by an automaton or that can be defined by a grammar. For instance, automatic structures,
defined for groups by Epstein et al.~\cite{epstein_wordproc} and generalized to semigroups by Campbell et
al.~\cite{campbell_autsg}, are a way of defining a group or semigroup using binary relations, recognizable by a
synchronous two-tape finite automaton that describe how generators for the group multiply normal-form words; in such
groups and semigroups fundamental questions like the word problem are solvable. Automatic presentations for relational
structures \cite{khoussainov_autopres,rubin_survey} similarly use synchronous multi-tape finite automata that recognize,
in terms of some regular language of representatives, the relations in the signature of a structure.

Possibly the most well-known type of relation in the application of language theory to algebra is the word problem,
which is a binary relation that relates pairs of words over a generating set for a semigroup that represent the same
element of that semigroup. In the literature, this binary relation has been studied in the context of language theory
from two perspectives. In one viewpoint, which we call the \emph{two-tape} viewpoint, an element $(u,v)$ of the binary
relation is thought of as being read (synchronously or asynchronously) by a two-tape automaton
\cite{brough_inverse,neunhoffer_deciding,pfeiffer_phd}. In the other, which we call the \emph{unfolded} viewpoint, the
element $(u,v)$ is represented by a word $u\#v^\rev$, where $\#$ is a new symbol and $v^\rev$ denotes the reverse of $v$
\cite{brough_inverse,holt_onecounter,holt_cfcoword}. The latter viewpoint, which can also be thought of in terms of reading
the first tape forwards and then the second tape in reverse, is a very natural representation for a binary relation when
a stack is involved.  In particular, hyperbolic groups in the sense of Gromov \cite{gromov_hyperbolic} can be
characterized using context-free languages \cite{gilman_wordhyperbolic} in a way that closely resembles this, and it is
this linguistic characterization that has given rise to the theory of word-hyperbolic semigroups
\cite{cm_wordhypunique,duncan_hyperbolic}, since the geometric definition of hyperbolicity is less natural for
semigroups. This in turn led to the study of semigroups with context-free word problem \cite{hoffmann_contextfree}.

These considerations motivate the present paper, which compares and contrasts the binary relations that can be defined
in the two-tape and unfolded perspectives for a number of language classes: rational ($\rat$), one-counter ($\OC$),
context-free ($\CF$), ET0L ($\etol$), EDT0L ($\edtol$), linear indexed ($\lin$), and indexed ($\ind$). These classes were chosen because
their applications to semigroups or groups have previously been studied; see
\cite{pfeiffer_phd,holt_onecounter,hoffmann_contextfree,brough_inverse,bcm_crosssection,holt_indexedcoword,ciobanu_solutions,ciobanu_applications,gilman_nested}.

\begin{figure}[t]
  \centering
  \begin{tikzpicture}[x=20mm,y=20mm]
    \useasboundingbox (-.5,.8) -- (6,-1.8);
    \begin{scope}[
      every node/.style={
        rectangle with rounded corners,
        draw=gray,
        font=\small,
        inner sep=.5mm,
      }
      ]
      \node (urat) at (0,0) {$U(\rat)$};
      \node (trat) at (0,-1) {$T(\rat)$};
      \node (uone) at (1,0) {$U(\OC)$};
      \node (tone) at (1,-1) {$T(\OC)$};
      \node (ucf) at (2,0) {$U(\CF)$};
      \node (tcf) at (2,-1) {$T(\CF)$};
      \node (ulin) at (3,.5) {$U(\lin)$};
      \node (tlin) at (3,-.25) {$T(\lin)$};
      \node (uetol) at (3,-.75) {$U(\etol)$};
      \node (tetol) at (3,-1.5) {$T(\etol)$};
      \node (uind) at (4,0) {$U(\ind)$};
      \node (tind) at (4,-1) {$T(\ind)$};
      \node[align=center] (utcsl) at (5.2,-0.5) {$U(\csl)$\\$=$\\$T(\csl)$};
    \end{scope}
    \draw[gray,rounded corners=1mm]
    ($ (uone.north) + (0,.13) $) -|
    ($ (ulin.west) + (-.13,0) $) |-
    ($ (ulin.north) + (0,.13) $) -|
    ($ (ulin.east) + (.13,0) $) |-
    ($ (tlin.south) + (0,-.13) $) -|
    ($ (tcf.east) + (.13,0) $) |-
    ($ (trat.south) + (0,-.13) $) -|
    ($ (trat.west) + (-.13,0) $) |-
    ($ (trat.north) + (0,.13) $) -|
    ($ (uone.west) + (-.13,0) $) |-
    ($ (uone.north) + (0,.13) $);
    \draw[contained] (urat) edge (uone);
    \draw[contained] (uone) edge (ucf);
    \draw[rcontained] (ucf) edge (ulin);
    \draw[contained] (ucf) edge (uetol);
    \draw[contained] (ulin) edge (uind);
    \draw[rcontained] (uetol) edge (uind);
    \draw[contained] (trat) edge (tone);
    \draw[contained] (tone) edge (tcf);
    \draw[contained] (tcf) edge (tetol);
    \draw[contained] (tlin) edge (tind);
    \draw[rcontained] (tetol) edge (tind);
    \draw[contained] (urat) edge (trat);
    \draw[contained] (uone) edge (tone);
    \draw[contained] (ucf) edge (tcf);
    \draw[contained] (ulin) edge (tlin);
    \draw[contained] (uetol) edge (tetol);
    \draw[maybecontained] (uind) edge (tind);
    \draw[incomparable] (trat.75) -- (uone.220);
    \draw[incomparable] (tone) -- (ucf);
    \draw[incomparable] (tcf) -- (uetol);
    \draw[incomparable] (tlin) -- (uetol);
    \draw[incomparable] (ulin.290) -- ($ (uind) + (-.4,0) $) -- ($ (tind) + (-.4,0) $) -- (tetol.70);
    \draw[rcontained] (trat) edge (ucf);
    \draw[rcontained] (tcf) edge (ulin);
    \draw[contained,dotted] (uind) edge[bend left=15] (utcsl);
    \draw[contained,dotted] (tind) edge[bend right=15] (utcsl);
  \end{tikzpicture}
  \caption{Illustration of the language hierarchy of relations. (Notation for classes of languages and relations is as
    defined in \S~\ref{sec:preliminaries}.)  Hooked arrows indicate proper inclusion.  Zigzag lines indicate
    incomparability.  The obvious incomparabilities between $U(\lin)$ and $U(\etol)$ and between $T(\lin)$ and
    $T(\etol)$ are not pictured, for clarity. Other relationships that are not either pictured or implied are currently
    unknown.  The single unhooked arrow from $U(\ind)$ to $T(\ind)$ indicates an inclusion not known to be proper. When
    restricted to relations over a one-symbol alphabet, the classes inside the grey outline coincide.}
  \label{hierarchy}
\end{figure}
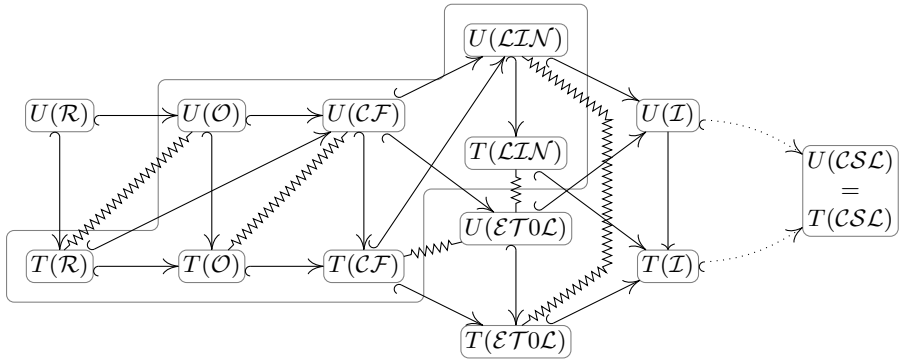

What emerges is the language hierarchy illustrated in Figure~\ref{hierarchy}. There is the expected straightforward
containment of classes of binary relations representable in the two-tape perspective, following the containment of
language classes, and similarly for those representable in the unfolded viewpoint. The two perspectives coincide for
context-sensitive languages ($\csl$). For each other language class, the class of unfolded binary relations is
contained in the corresponding class of two-tape binary relations, with the containments being proven to be proper in
each case except for $\ind$, where it remains an open question. Some of the two-tape classes are contained in
unfolded classes corresponding to larger language classes. For instance, the two-tape $\CF$ relations are contained
in the unfolded $\lin$ relations (but not the unfolded $\etol$ relations). In other cases, we have proven
incomparability of some classes of relations.

For several of these classes (inside the grey outline in Figure~\ref{hierarchy}), the subclasses of binary relations
between words over a $1$-symbol alphabet coincide. It remains open whether the analogous subclasses of other classes
coincide like this. Note that all the witnesses we give for proper containment and incomparability are over alphabets of
at most $2$ symbols, so there is no question of coincidence of subclasses of relations over $2$-symbol alphabets.




\section{Preliminaries}
\label{sec:preliminaries}

Throughout the paper, $X$ will denote a finite alphabet and $\emptyword$ the empty word.

There are two ways to describe a binary relation $\rho$ on $X^*$ (that is, a subset $\rho$
of $X^* \times X^*$) using languages:
\begin{itemize}
\item one can specify a sublanguage $L_\rho$ of $X^*\#X^*$ (where $\#$ is a symbol not in $A$) such that
  $L_{\rho} = \gset{u\#v^\rev}{u,v\in A^*, u\mathrel{\rho} v}$;
\item one can specify a language $L$ over the alphabet
  $\Xtt := (X \cup \set{\emptyword}) \times (X \cup \set{\emptyword})$ such that $\rho = L\pi$, where the map
  $\pi$ is defined by
  \[
    (x_1,y_1)\cdots (x_k,y_k) \mapsto (x_1\cdots x_k,y_1\cdots y_k) \quad \text{for $x_i,y_i \in X \cup \set{\emptyword}$.}   
  \]
  (Note that there may be many choices for the language $L$, and multiple words within $L$ mapping
  onto a given element of $\rho$.)
\end{itemize}
For  a class of languages $\C$, let
\[
  T(\C) = \gset{\rho}{\rho = L\pi \text{ for some $L \in \C$}}
\]
and
\[
  U(\C) = \gset{\rho}{L_\rho \in \C}.
\]
Mnemonically, $T(\cdot)$ signifies `(two-)tape relation'; $U(\cdot)$ signifies `unfolded relation'. In the remainder of
the paper, we will often omit mention of the map $\pi$ and simply think of grammars over $\Xtt$ as defining binary relations.

\begin{example}\label{rev}
Let $|X|\geq 2$ and let $\rho = \gset{ (w,w^\rev)}{ w\in X^*}$ be the relation that relates each word
to its reverse.  Then $\rho$ is in $T(\CF)$ but not in $U(\CF)$.
\end{example}
\begin{proof}
The following context-free grammar defines a language $L$ over $\Xtt$ with $L\pi = \rho$.
\[
S\ra (x,\emptyword) S (\emptyword,x) \mid \emptyword \quad (\forall x\in X).
\]

However, $L_\rho = \gset{ w\# w}{w\in X^*}$, which is well known not to be context-free
(this can be proved by the pumping lemma).
\qed
\end{proof}

It is clear that if $\C$ and $\D$ are classes of languages and $\C \subseteq \D$, then $U(\C) \subseteq U(\D)$ and
$T(\C) \subseteq T(\D)$.  The following result implies that \emph{proper} containment is also inherited from language classes
by classes of two-tape and unfolded relations. The technical condition that language classes considered are closed under
left quotient and left concatenation by a single symbol is a very weak property that is satisfied by all language
classes considered in this paper and many others.

\begin{proposition}\label{basics}
  Let $\C$ and $\D$ be classes of languages closed under homomorphism and under left quotient and left concatenation by
  a single symbol.  If $\C\setminus \D$ is non-empty, then the following classes of relations are all non-empty:
  $U(\C)\setminus U(\D)$, $T(\C)\setminus T(\D)$, $U(\C)\setminus T(\D)$ and $T(\C)\setminus U(\D)$.
\end{proposition}

\begin{proof}
For any language $K\subseteq X^*$, let $\rho_K = \gset{ (\emptyword, w)}{w\in K}$.
Then $L_{\rho_K} = \gset{\#w}{w\in K}$ is in $\C$ (and hence $\rho_K\in U(\C)$)
if and only if $K\in \C$.
 If $\rho\in T(\C)$ then the projections of $\rho$ onto each component are languages in $\C$,
 so we also have $\rho_K\in T(\C)$ if and only if $K\in \C$.
 Hence if $K\in \C\setminus \D$, then $\rho_K$ is in $U(\C)\setminus U(\D)$,
 $T(\C)\setminus T(\D)$, $U(\C)\setminus T(\D)$ and $T(\C)\setminus U(\D)$.
 \qed
\end{proof}

Since even a linear bounded automaton is powerful enough that there is 
no effective difference between the input $(u,v)$ and the input $u\#v^\rev$,
the following proposition is straightforward.

\begin{proposition}\label{csl}
If $\C$ is the class of context-sensitive, recursive or
recursively enumerable languages, we have $U(\C) = T(\C)$.
\end{proposition}

Note that there are two definitions of linear indexed grammars in the
literature.  The one we have in mind and denote by $\lin$ is that in which,
at each derivations step, flags are only copied to one of the 
(possibly multiple) non-terminals on the right-hand side of each production.
The linear indexed grammars in this sense are
equivalent to tree-adjoining grammars \cite[p.72]{kallmeyer_parsing}.
The second definition (see for example \cite{duske_linear}) requires there to
be at most one non-terminal on the right hand side of any production.
All of our results stated for $\lin$ in fact hold equally for the second definition.

\section{Some further examples}

Before moving on to establishing the hierarchy illustrated in Figure~\ref{hierarchy},
we give several examples of binary relations and the language classes they belong to.
While these examples do illustrate various aspects of the hierarchy, the intention 
in this section is not to prove these points but to explore some of the richness of 
binary relations from a language-theory viewpoint.  The hierarchy will be systematically 
established in sections~\ref{UT} and~\ref{TU}, with an attempt to give simple examples 
wherever possible.

\begin{example}\label{functions}
Any function $f:\N_0\ra \N_0$ can be viewed as a binary relation over an alphabet $|X| = \{x\}$,
by defining $\rho_f  = \gset{ (x^n, x^{f(n)})}{n\in \N_0}$.
Consider the following functions:
\begin{itemize}
\item For $e: \N_0\ra \N_0$ given by $n\mapsto n$, we have $\rho_e\in T(\rat)\setminus U(\rat)$.
\item For $f: \N_0\ra \N_0$ given by $n\mapsto n \bmod p$ (some $p\in \N$), 
we have $\rho_f\in U(\rat)$.
\item For $g: \N_0\ra \N_0$ given by $n\mapsto n^2$, we have $\rho_g\in U(\etol)\setminus T(\lin)$.
\item For $h: \N\ra \N$ given by $n\mapsto n^n$, we have $\rho_h\in U(\csl)\setminus T(\ind)$.
\end{itemize}
\end{example}

\begin{proof}
The relation of the identity function $e$ is 
$\rho_e = \gset{ (x^n, x^n)}{n\in \N_0} = (x,x)^*$, which is rational.
However, $L_{\rho_e} = \gset{ x^n\# x^n}{n\in \N_0}$, which is easily
seen by the pumping lemma to be non-rational.

We have $L_{\rho_f} = x^n \# x^{n \bmod p}$.  This is recognised by a finite automaton 
with states $\{q_0,\ldots,q_{p-1}, r_0,\ldots, r_{p-1}\}$, with $q_0$ the initial state,
$r_0$ the unique final state, and transitions 
\[ q_i \stackrel{x}{\ra} q_{i+1 \bmod p}, \quad q_i \stackrel{\#}{\ra} r_i, \quad r_j \stackrel{x}{\ra} r_{j-1}, \]
for $0\leq i\leq p-1$ and $1\leq j\leq p-1$.

An ET0L-system for the language $L_{\rho_g}$ can be obtained by considering the 
expression of $n^2$ as the sum of the first $n$ odd numbers.
The system has axiom $A\#B$ and tables
\begin{itemize}
\item Table 1:  $A\ra xA, \quad B\ra xBC, \quad C\ra x^2 C$;
\item Table 2:  $A,B,C\ra \emptyword.$
\end{itemize}
At each iteration of Table~1, the number of occurrences of $A$ and $B$ remains at one each,
while the number of occurrences of $C$ is increased by $1$.
One $x$ is generated to the left of $\#$ by $A$, while $2c+1$ symbols $x$ are generated to the right of $\#$, where $c$ is the number of occurrences of $C$ in the current sentential form.  
Hence after $n$ iterations of Table~1 and one application of Table~2, we obtain
$x^n\# x^{s(n)}$, where $s(n) = \sum_{i=0}^{n-1} 2i+1 = n^2$. 
Thus $\rho_g\in U(\etol)$.
The Parikh image of a linear indexed language is a semilinear set
\cite[Theorem~5.1]{duske_linear}, hence no language
$L$ with $L\pi = \rho_g$ can be linear indexed, so $\rho_g\notin U(\lin)$.

Suppose $\rho_h\in T(\ind)$; then the language $\gset{x^{n^n}}{n\in \N}$ (obtained by 
projecting onto the second tape) is indexed.
By Gilman's shrinking lemma \cite{gilman_shrinking} (with $m=1$), there exists $k$ such that 
if $n^n>k$ then there exist $r\in \set{2,\ldots k}$ and $l_i\in \N_0$ such that
we can write $n^n = \sum_{i=1}^r l_i$, and
for any $j\in \set{1,\ldots,r}$ there is a proper subset $I$ of $\set{1,\ldots,r}$ with $i\in I$
and $s_j := \sum_{i\in I} l_i = t^t$ for some $t\in \N$.
We may choose $j$ such that $l_j\geq n^n/k$, since $r\leq k$.
But if $n>k$, we then have $l_j\geq n^{n-1} > (n-1)^{(n-1)}$, so that it is not possible 
for the sum $s_j$ to be of the form $t^t$ for $t\in \N$.  Hence $\rho\notin T(\ind)$.
However, $L_{\rho_h}$ is recognised by a linear bounded automaton that writes input
of the form $x^n\#x^m$ onto the tape (rejecting input not of this form) and then checks
whether $m = n^n$ by essentially performing successive multiplications by $n$, using
the initial $x^n$ portion of the tape both to count to $n$ for the multiplication, and to ensure
it is performed $n$ times.  Hence $\rho\in U(\csl)$.
\qed
\end{proof}

The method of proof for $\rho_g\in U(\etol)$ should be extendable to all functions 
$n\mapsto n^k$ for $k\in \N$ (the authors have shown it for $k=3$), and we conjecture 
more strongly:

\begin{conjecture}\label{polynomials}
Let $p:\N\ra \N$ be a polynomial function.  Then $\rho_p\in U(\etol)$.
\end{conjecture}

We note that a fact in some sense `inverse' to Conjecture~\ref{polynomials} is known:
If $p:\N\ra \N$ is a polynomial function, then there exists an D0L language (D0L is a subclass of ET0L)
with growth function $p$ \cite{ruohonen_synthesis}.

\begin{example}\label{content}
Let $X = \set{1,\ldots,n}$ and let
\[o = \gset{(w, 1^{|w|_1}\ldots n^{|w|_n})}{w\in X^*}.\]
The relation $o$ is a function that sorts a word into order.
The language complexity of $o$ appears to increase as $n$ increases:
$o$ is in $T(\CF)$ for $n\leq 2$, in $U(\lin)$ for $n\leq 3$,
in $T(\lin)$ for $n\leq 4$, and in $U(\ind)$ for all $n$.
\end{example}
\begin{proof}
We need only provide grammars for the largest $n$ claimed in each case,
since the relations for smaller $n$ are homomorphic images of those for larger $n$.

For $n=2$, a context-free grammar for $o$ is given by
\[ S\ra (1,1)S \mid (2,\emptyword)S(\emptyword,2)\mid \emptyword.\]

In indexed grammars we will be using the notation $A^f$ for a non-terminal
with flag $f$.  Thus in the following grammars the superscripts in $\N$ denote flags 
rather than powers of the non-terminal.

For $n=3$, a linear indexed grammar for $L_o$ is given by
\begin{align*}
S&\ra 1S1 \mid 2S^2 \mid 3S^3 \mid \#T\\
T^2&\ra T2\\
T^3&\ra 3T\\
T^{\$}&\ra \emptyword.
\end{align*}
Hence $o\in U(\lin)$ for $n\leq 3$, since $\lin$ is closed under homomorphism.

For $n=4$, a linear indexed grammar for $o$ is given by
\begin{align*}
S&\ra (1,1)S\mid (2,\emptyword)S^2\mid (3,\emptyword)S^3\mid (4,\emptyword)S(\emptyword,4)
\mid T\\
T^2&\ra (\emptyword,2)T\\
T^3&\ra T(\emptyword,3)\\
T^{\$}&\ra \emptyword.
\end{align*}
Hence $o\in T(\ind)$ for $n\leq 4$.

An indexed grammar for $L_o$ is given by the following productions,
where $x$ ranges over all elements of $X$:
\begin{align*}
S&\ra S^x \mid T\# T_1 T_2 \ldots T_n
& T^x&\ra xT\\
T_x^x&\ra T_x
& T^{\$}&\ra \emptyword\\
T_x^y&\ra T_x \quad (\forall y\in X\setminus\{x\})
& T_x^{\$}&\ra \emptyword .
\end{align*}
Hence $o\in U(\ind)$ for all $n$.
\qed
\end{proof}

\begin{proposition}
The relation $o$ in Example~\ref{content} is not in $T(\CF)$ for $n>2$.
\end{proposition}

\begin{proof}
Suppose $n\geq 3$ and let $L$ be the language obtained by intersecting 
$o$ with $(123)^*\times X^*$ and then projecting onto the second tape.
If $o\in T(\CF)$, then $L$ is context-free.  But $L = \gset{1^k2^k3^k}{k\in \N_0}$,
which is not context-free.
\qed
\end{proof}

\begin{conjecture}
The relation $o$ in Example~\ref{content} is not in $T(\etol)$ or in $T(\lin)$ for sufficiently large $n$.
\end{conjecture}

\begin{proposition}\label{cyclic}
The cyclic permutation relation $\kappa = \gset{ (uv, vu)} {u,v\in X^*}$ is in
$U(\etol)$ and in $U(\lin)$, but not in $T(\CF)$ for $|X|\geq 2$.
\end{proposition}

\begin{proof}
If $|A|=1$ then $\kappa$ is simply the equality relation, so assume $|X|\geq 2$.

An ET0L-system for $L_\kappa = \gset{uv\#u^\rev v^\rev}{u,v\in X^*}$ with axiom $UV\#U'V'$ is given by the following tables:
\begin{itemize}
\item
Table $1x$ (for each $x\in X$):
$U\ra xU$, $U'\ra U'x$;
\item
Table $2x$ (for each $x\in X$):
$V\ra xV$, $V'\ra V'x$;
\item
Table 3:
$U\ra \emptyword$, $V\ra \emptyword$.
\end{itemize}

A linear indexed grammar for $L_\kappa$ with start symbol $S$ is given by the productions:
\begin{align*}
S&\ra xS_x \; (\forall x\in X) \mid A;\\
A&\ra xAx \; (\forall x\in X) \mid \#B;\\
B_x&\ra xB \; (\forall x\in X);\\
B_{\$}&\ra \emptyword.
\end{align*}

To show that $\kappa\notin T(\CF)$, we first assume $|X|\geq 6$, and later explain
how to adapt the proof of this case for all $|X|\geq 2$.

Suppose that $\kappa\in T(\CF)$ and that $a,b,c,x,y,z$ are distinct elements of $X$.
Then the relation $\kappa_1 = \kappa \cap (a^*b^*c^*x^*y^*z^*, \, x^*y^*z^*a^*b^*c^*)$
is also in $T(\CF)$.  Let $Y = \set{a,b,c,x,y,z}$ and let $L\subseteq \Ytt$ be a context-free
language such that $L\pi = \kappa_1$.
Let $k$ be the pumping constant for $L$.
For $n>k$, let $u = a^n b^n c^n$ and $v_n = x^n y^n z^n$ and
let $w$ be any word in $L$ such that $w\pi = (uv,vu)$.

It will be helpful to view $w$ as a shuffle of the strings $u_1v_1$ and $u_2v_2$,
where $u_1$ is the string mapping onto $u$ on the first tape, and $u_2,v_1,v_2$
are defined analogously.  We shall be concerned with two (possibly empty)
special subwords of $w$:
these are $e_u$, the subword beginning with the first symbol from $u_2$ and ending with the
last symbol from $u_1$; and $e_v$, the subword beginning with the first symbol from
$v_1$ and ending with the last symbol from $v_2$.  The length of these subwords
shows how much the two occurrences of $u$
respectively $v$ overlap within $w$.

Since $|w|>k$, we can write
$w = pqrst$ such that $|r|<k$ and $pq^irs^it\in L$ for all $i\in \N_0$,
and at least one of $q$ or $s$ is non-empty.

Let $q\pi = (q_1,q_2)$ and $s\pi = (s_1,s_2)$.
Since $q$ and $s$ can be pumped in $w$, the form of pairs in $\kappa_1$
implies that each $q_i$ and $s_i$ is a (possibly trivial) power of a single letter from $Y$.
Moreover, at least two these powers must be non-trivial, otherwise deleting $q$ and $s$
from $w$ would result in a string $w'$ with different lengths on each tape, contradicting
$w'\pi\in \kappa_1$.

Let $P$ denote the subset of $\{q_1,q_2,s_1,s_2\}$ consisting of non-empty strings.
We proceed by case analysis, depending on the size of $P$, and whether
the elements of $P$ are in $A = a^*\cup b^*\cup c^*$ or $Z = x^*\cup y^*\cup z^*$ or both.\\

\textit{Case 1}  $|P|=2$.
To maintain membership of $L$ upon pumping, the strings in $P$ must
be in the same set $A$ or $Z$.  They also must come from different tapes (so
have different subscripts).
Suppose first that $P = \set{q_1,q_2}$.
If $q_1,q_2\in A$, then $v_2$ is a substring of $p$,
while $v_1$ is a substring of $rst$.  Thus $e_v$ is empty.
Symmetrically, if $q_1,q_2\in Z$, then $e_u$ is empty;
and if instead $P = \set{s_1,s_2}$, then
$e_u = \emptyword$ if $P\subseteq Z$, and $e_v = \emptyword$ if $P\subseteq A$.

Now suppose $P = \set{q_1,s_2}\subseteq Z$.
Then $u_1$ and $u_2$ are contained in $t$ and $p$ respectively, so $e_u=\emptyword$.
Symmetrically, if $P = \set{q_2,s_1}\subseteq A$ then $e_v = \emptyword$.

However, if $P = \set{q_1,s_2}\subseteq A$, we do not obtain separation.
In this case $v_1$ is contained in $rst$,
while $v_2$ is contained in $pqr$.
Thus these strings overlap only in $r$.
Thus $|e_v|\leq |r|<k$.
Symmetrically, in the final subcase $P = \set{q_2,s_1}\subseteq Z$,
we have $|e_u|<k$.\\

\textit{Case 2}  $|P|=3$.
In this case, all the strings in $P$ must still be in the same set $A$ or $Z$,
since on one of the tapes only one string (a power of a single letter) is being pumped.
Moreover, $P$ must always contain one of the pairs $\set{q_1,q_2}$ or
$\set{s_1,s_2}$, so that this is actually a more restrictive situation than Case 1
and always leads to either $e_u$ or $e_v$ being empty.\\

\textit{Case 3}  $|P|=4$.
If all elements of $P$ are in the same set $A$ or $Z$, then this is
even more restricted than Case 2 and we have $e_u$ or $e_v$ empty.
However, it is possible for $P$ to contain strings from both $A$ and $Z$.
The only situation in which this occurs is $q_1,s_2\in A$ and $q_2,s_1\in Z$.
In this case, $u_1$ and $v_2$ are both contained in $pqr$, while
$u_2$ and $v_1$ are both contained in $rst$.
Hence $e_u$ and $e_v$ are both substrings of $r$ and have length less than $k$.\\

Thus we have established that in all cases either $|e_u|<k$ or $|e_v|<k$.
The argument was valid for all $n>k$, and so at least one of the following
holds: for infinitely many $n$, we have $|e_u|>k$; or for infinitely many $n$, we have $|e_v|>k$.

Suppose first that $|e_u|>k$ for infinitely many $n$.
Let $\theta$ be the homomorphism given by deleting all
occurrences of symbols from $\set{c,x,y,z}$ on either tape,
and let $L' = L\theta$.
Then $L'$ contains $\alpha = (a,\emptyword)^n (b,\emptyword)^n (\emptyword,a)^n (\emptyword,b)^n$
for arbitrarily large $n$, since any `intermingling' between
symbols from $A\times \set{\emptyword}$ and $\set{\emptyword}\times A$
in words with $|e_u|<k$ occurs only among the symbols $(\emptyword,c)$ and
$(a,\emptyword)$.
Taking $n$ greater than the pumping constant of $L'$
(which, as a homomorphic image of $L$, is context-free),
there exist two substrings of $w_n$ that can be simultaneously pumped,
and these substrings must be contained within one of the subwords
$(a,\emptyword)^n (b,\emptyword)^n$, $(b,\emptyword)^n (\emptyword,a)^n$
or $(\emptyword,a)^n (\emptyword,b)^n$.  But this implies that $L'$ contains words
that cannot be the image under $\theta$ of words in $L$,
contrary to the definition of $L'$.
Similarly, taking if $|e_v|>k$ for infinitely many $n$, then
taking a homomorphism $\theta'$ that deletes all occurrences of $a,b,c$ and $x$ again results in a
non-context-free language.  Thus $L$ itself cannot have been context-free
and so $\kappa\notin T(\CF)$.

In general, for $|X|\geq 2$ we can modify the above proof for $|X|\geq 6$ as follows.
Let $0,1\in X$ and define $a=1, b=10, c=100, x=1000, y=10000$ and $z=100000$.
The entire argument given above holds, with the modification that we must replace
the homomorphisms $\theta$ and $\theta'$ by rational transductions in order to delete
$c,x,y,z$ (resp.\ $a,b,c,x$).
\qed
\end{proof}

\section{Comparing $U(\C)$ and $T(\C)$}\label{UT}

Although the following result is a special case of Proposition~\ref{indexed},
we give its proof separately as a useful `warm-up' for the proof of the later
proposition.

\begin{proposition}\label{rat}
$U(\rat)$ is properly contained in $T(\rat)$.
\end{proposition}

\begin{proof}
Let $X$ be a finite set and $\rho$ a relation on $X^*$ with $L_{\rho}$ rational.  Let
$\Gamma = (N,X\cup \{\#\},P,S\})$ be a left-regular grammar for $L_\rho$.
We may assume that $N$ is partitioned into
sets $N_1, N_2, N_3$ such that all productions in $P$ have one of the following forms:
\begin{align*}
A &\ra aB, && \text{where $A\in N_1$, $B\in N_1\cup N_2$, $a\in X$;}\\
A &\ra \#, \quad A\ra \#B &&\text{where $A\in N_2$, $B\in N_3$;}\\
A &\ra a, \quad A\ra aB && \text{where $A,B\in N_3$, $a\in X\cup \set{\emptyword}$.}
\end{align*}
Let $\Gamma'$ be the left-regular grammar obtained by replacing productions as follows.  For $A_i\in N_i$, replace
$A_1\ra aB$ by $A_1\ra (a,\emptyword)B$, $A_2\ra \#$ by $A_2\ra \emptyword$, $A_2\ra \#B$ by $A_2\ra B$, $A_3\ra a$ by
$A_3\ra (\emptyword,a)$ and $A_3\ra aB$ by $A_3\ra Ba$.  Then $\Gamma'$ generates exactly the set $\gset{(u,v)}{u\#v^\rev\in L_\rho}$, which is $\rho$, and so $\rho$ is rational.
An example of a relation in $T(\rat)\setminus U(\rat)$ is the relation $\rho_e$ in Example~\ref{functions}.
\qed
\end{proof}

\begin{proposition}\label{1c}
$U(\OC)$ is properly contained in $T(\OC)$.
\end{proposition}

\begin{proof}
  Let $\rho$ be a binary relation on $X^*$ that lies in $U(\OC)$. Consider a one-counter automaton $\A$ that
  recognizes $L_\rho$. View $\A$ as having an integer counter that starts at $0$ and which it increments or
  decrements as it reads each symbol of a word $u\#v^\rev$, accepting if the counter has value $0$ at the end of the
  input. Build a one-counter automaton $\B$ over $X^2_\emptyword$ recognizing
  $\gset{(u,\emptyword)(\emptyword,v)}{(u,v) \in \rho}$ that functions as follows. It first reads all input from
  its first tape, simulating $\A$ on this input followed by $\#$. It then reads the input from its second tape,
  nondeterministically simulates $\A$ in reverse on this tape (so that increments to the counter become decrements,
  and vice versa). Since increments and decrements to the counter commute, it is clear that
  $\B$ accepts $(u,v)$ if and only if $\A$ accepts $u\#v^\rev$.

  An example of a relation in $T(\OC)$ but not in $U(\OC)$ (since it is not in $U(\CF)$) is given in
  Proposition~\ref{wp1ccf2}.  \qed
\end{proof}

\begin{proposition}\label{indexed}
Let $\mathfrak{C}$ be the class of indexed, linear indexed, ET0L, EDT0L, context-free
or rational languages. Then $U(\mathfrak{C})\subseteq T(\mathfrak{C})$.
\end{proposition}

\begin{proof}
  Let $\rho$ be a binary relation on $X^*$ that lies in $U(\ind)$ and let $\Gamma$ be an indexed grammar for $L_\rho$.
  We may assume that the set of non-terminals $N$ is partitioned into sets $N_L, N_{\#}, N_R$ such that all productions
  are of one of the types in the table below, where we adopt the convention that a nontermal with a subscript
  $H \in \set{L,\#,R}$ is in $N_H$, and the convention that a Greek letter with a subscript $H$ denotes a word over
  $N_H$ (potentially with flags).

  We construct a new indexed grammar $\Gamma'$ with non-terminals
  $N' = \{A' \mid A\in N\}$ and start symbol $S'_{\#}$ (where $S$ is the start symbol of $\Gamma$) in which each
  production is replaced by the corresponding production shown below, and where $\alpha'$ denotes the word $\alpha$ with
  each non-terminal $A$ replaced by $A'$, and each terminal $a$ replaced by $(a,\emptyword)$ if $\alpha\in N_L^*$ and by
  $(\emptyword,a)$ if $\alpha\in N_R^*$.  A flag $f$ in brackets denotes that the flag may or may not be
  present in the production.\\

{%
  \smallskip
  \centering
  \begin{tabular}{rcl@{\quad}lrcl}
    \toprule
            \multicolumn{3}{l}{Production in $\Gamma$} & \multicolumn{4}{l}{Corresponding Production in $\Gamma'$} \\
    \midrule
    $A_{\#}(f)$ & $\ra$ & $\alpha_L \# \gamma_R$      & & $A'_{\#}(f)$ & $\ra$ & $ \alpha'_L (\gamma'_R)^\rev$         \\
    $A_{\#}(f)$ & $\ra$ & $\alpha _L B_{\#} \gamma_R$ & \kern 9mm & $A'_{\#}(f)$ & $\ra$ & $\alpha'_L (\gamma'_R)^\rev B'_{\#} $ \\
    $A_L(f)$    & $\ra$ & $\alpha_L$                  & & $A'_L(f)$    & $\ra$ & $\alpha'_L$                           \\
    $A_R(f)$    & $\ra$ & $\alpha_R$                  & & $A'_R(f)$    & $\ra$ & $(\alpha'_R)^\rev$                    \\
    $A_H$ & $\ra$ & $B_H f$				& & $A'_H$ & $\ra$ & $B'_H f$,\\
    \bottomrule
  \end{tabular}
  \par
  \smallskip
}

\noindent
At any stage in a derivation in $\Gamma$, the current string has the form
$\alpha_L B \beta_R$, where $\alpha_L\in (N_L\cup X)^*$,
$B\in N_{\#}\cup \{\#\}$ and $\beta_R\in (N_R\cup X)^*$.
Since the non-terminals in $N_L$ and $N_R$ produce terminals
on the first and second tape respectively, they commute with each other
(that is, $A_L B_R$ and $B_R A_L$ both produce the same language).
Thus at the same point in the corresponding derivation in $\Gamma'$,
the current string will be equivalent (in the sense that it generates the same language)
to $\alpha'_L (\beta'_R)^\rev B'_{\#}$ or to $\alpha'_L (\beta'_R)^\rev$.
Thus the strings produced by $\Gamma'$ are exactly those equivalent to
$\alpha'_L \beta'_R$ for some $\alpha_L\#(\beta_R)^\rev\in L_\rho$.
Since for $\alpha_L,\beta_R\in X^*$ we have $\alpha'_L = (\alpha,\emptyword)$
and $(\beta')^\rev_R = (\emptyword, \beta)$, we conclude that $\Gamma'$ generates
the relation $\rho$ and hence $\rho\in T(\ind)$.

Note that the transformation from $\Gamma$ to $\Gamma'$ preserves the property
of being a linear indexed, E(D)T0L-indexed, context-free or rational grammar.
Hence the containment of $U(\mathfrak{C})\subseteq T(\mathfrak{C})$
also holds for all the classes mentioned.
\qed
\end{proof}

\begin{proposition}\label{CF}
The containment of $U(\CF)$ in $T(\CF)$ is proper.
\end{proposition}
\begin{proof}
See Example~\ref{rev}.  \qed
\end{proof}

\begin{proposition}\label{linind}
The containment of $U(\lin)$ in $T(\lin)$ is proper.
\end{proposition}

\begin{proof}
  Let $\rho = \gset{(a_1^n a_2^n a_3^n, b_1^n b_2^n b_3^n)}{n\in \N_0}$.  Then $\rho$ is obtained from the language
  $L_3 = \gset{ a_1^n a_2^n a_3^n}{n\in \N_0}$ by replacing each $a_i$ by $(a_i,b_i)$.  Since $L_3$ is linear indexed by
  \cite[Problem~4.1]{kallmeyer_parsing} (recalling that the classes of languages defined by linear indexed and
  tree-adjoining grammars coincide), $\rho$ lies in $T(\lin)$.  But
  $L_\rho = \gset{ a_1^n a_2^n a_3^n \# b_3^n b_2^n b_1^n}{n\in \N_0}$ has as a homomorphic image the language
  $L_5 = \gset{a^nb^nc^nd^ne^n}{n\in \N_0}$, which is not linear indexed \cite[Problem~4.5]{kallmeyer_parsing}.  Hence
  $\rho\notin U(\lin)$.  \qed
\end{proof}

\begin{proposition}\label{et0l}
  The containment of $U(\etol)$ in $T(\etol)$ is proper, and moreover there exists a relation in
  $T(\CF)\setminus U(\etol)$.
\end{proposition}

\begin{proof}
  By Proposition~\ref{indexed}, $U(\etol) \subseteq T(\etol)$. Let $K\subseteq X^*$ be a context-free language that is not
EDT0L (an example on a $2$-letter alphabet exists
\cite{ehrenfeucht_onsome}).
The language $\gset{w\# w^\rev}{w\in K}$,
which is $L_\sigma$ for the relation $\sigma = \gset{(w,w)}{ w\in K}$,
is not ET0L \cite{ehrenfeucht_relationship}.
However, a context-free grammar for $\sigma$ can be obtained
by replacing every output symbol $x$ in a context-free grammar for $K$
by $(x,x)$.
\qed
\end{proof}

\begin{proposition}\label{edt0l}
$U(\edtol) =  T(\edtol)$.
\end{proposition}
\begin{proof}
By Proposition~\ref{indexed}, it suffices to prove that
$T(\edtol)\subseteq U(\edtol)$.
Let $\rho$ be a binary relation on $X^*$ that lies in $T(\edtol)$ and let
$H = (V,\Xtt,\Delta,I)$ be an EDT0L-system for $\rho$.
Define an ET0L-system $H' = (V_1\cup V_2,X\cup\{\#\},\Delta',I')$
with $V_1$, $V_2$ copies of $V$, $I' = I_1 \# I_2$ and
$\Delta'$ consisting of a set of tables in one-to-one
correspondence with the tables in $\Delta$,
with each production $A\ra \alpha$, $\alpha\in (V\cup \Xtt)^*$
being replaced by two productions $A_1\ra \alpha_1$
and $A_2\ra \alpha_2^\rev$, where $\alpha_i$ is the same
as $\alpha$ except with each non-terminal in $V$ replaced
by the corresponding non-terminal in $V_i$, and the terminal
symbols replaced by their $i$-th component.
Since each table in $\Delta$ contains exactly one production
from $A$ for each $A\in V$, applying a given sequence of
tables in $\Delta'$ to $I'$ produces the word $u\# v^\rev$,
where $(u,v)$ is the word produced by applying the corresponding
sequence of tables in $\Delta$ to $I$.  Note that $\Delta'$ also
contains only one production from each non-terminal in $V_1\cup V_2$.
Hence $H'$ is an EDT0L-system for $L_\rho$.
\qed
\end{proof}

\begin{conjecture}\label{indconj}
$U(\ind)$ is properly contained in $T(\ind)$.
\end{conjecture}

\begin{example}\label{samelensamesymbolcount}
  Let $|X|\geq 2$ and for $a\in X$ define a relation $\rho$ on $X^*$
  by $u \mathrel\rho v$ iff $u$ and $v$ have the same length and $|u|_a=|v|_a$.
  This relation is in $T(\OC)$, since an automaton simply reads
  pairs of symbols (thus ensuring that two accepted words have the same length) and tracks the difference in the number
  of symbols $a$ on its two input tapes.  Thus, by Proposition~\ref{cflin} below,
  \[
    L_\rho = \gset{u\#v}{|u| = |v|, |u|_a = |v|_a}
  \]
  is a linear indexed language. Since $\lin$ is closed under homomorphisms, the relation
  $\sigma = \gset{(uv,uv)}{|u| = |v|, |u|_a = |v|_a}$, obtained by applying the homomorphism defined by $x \mapsto (x,x)$
  (for all $x \neq \#$) and $\# \mapsto \emptyword$, is also in $T(\lin)$ and so in $T(\ind)$. However, it is unlikely
  that $\sigma$ is in $U(\ind)$, since recognizing
  \[
    L_\sigma = \gset{uv\#v^\rev u^\rev}{|u| = |v|, |u|_a = |v|_a}
  \]
  seems to require two blind counters operating independently, which in turn does not seem to be workable with the nested stack automaton model of indexed languages.
(Two independent non-blind counters are of course equivalent to a Turing machine.)  
   Thus $\sigma$ is a potential witness for the proper containment of
  $U(\ind)$ in $T(\ind)$.
\end{example}

\section{Comparing $U(\C)$ and $T(\D)$ for $\D$ a subclass of $\C$}\label{TU}

Throughout this section, we will make use of Proposition~\ref{basics}
-- in particular the fact that if $\C\setminus \D\neq \emptyset$ then
$U(\C)\setminus T(\D)\neq \emptyset$ -- without further comment.

\begin{proposition}\label{ratcf}
$T(\rat)$ is a proper subclass of $U(\CF)$.
\end{proposition}
\begin{proof}
Let $\rho$ be a binary relation on $X^*$ that lies in $T(\rat)$ and let
$\Gamma = (N,\Xtt,P,S)$ be a left rational two-tape grammar for $\rho$
Define a context-free grammar
$\Gamma' = (N',X,P',S')$ whose
productions are derived from $P$ as follows (for $A,B \in N$ and $a \in X$):

{%
  \smallskip
  \centering
  \begin{tabular}{lrcl@{\quad}lrcl}
    \toprule
    \multicolumn{4}{l}{Production in $\Gamma$} & \multicolumn{4}{l}{Corresponding Production in $\Gamma'$}                                                            \\
    \midrule
 & $A$ & $\ra$ & $(a,\emptyword)B$ & \kern 15mm & $A'$ & $\ra$ & $aB'$ \\
 & $A$ & $\ra$ & $(\emptyword,a)B$ &  & $A'$ & $\ra$ & $B'a$ \\
 & $A$ & $\ra$ & $(a,\emptyword)$  &  & $A'$ & $\ra$ & $a\#$ \\
 & $A$ & $\ra$ & $(\emptyword,a)$  &  & $A'$ & $\ra$ & $\#a$ \\
 & $A$ & $\ra$ & $\emptyword$      &  & $A'$ & $\ra$ & $\#$, \\
    \bottomrule
  \end{tabular}
  \par
  \smallskip
}

It is easy to see that $L(\Gamma') = \gset{u\#v^\rev}{(u,v)\in \rho} = L_\rho$.

An example of a relation that lies in $U(\OC) \setminus T(\rat)$ and thus in $U(\CF) \setminus T(\rat)$ is given in
Proposition~\ref{wpratcf}.  \qed
\end{proof}

\begin{proposition}\label{rat1c}
The classes $T(\rat)$ and $U(\OC)$ are incomparable.
\end{proposition}

\begin{proof}
This follows from Proposition~\ref{basics} and Proposition~\ref{wpratcf} below.
\qed
\end{proof}

\begin{proposition}\label{1ccf}
The classes $T(\OC)$ and $U(\CF)$ are incomparable.
\end{proposition}
\begin{proof}
Let $X=\set{a,b}$ and $\rho = \gset{ (a^nb^n,b^na^n)}{n\in \N}$.
Then $\rho$ is one-counter, but $L_\rho = \gset{a^nb^n\#a^nb^n}{n\in \N}$
is not context-free.
An example of a relation in $U(\CF)\setminus T(\OC)$ is given in
Proposition~\ref{wp1ccf1}.
\end{proof}

\begin{corollary}\label{cfet0l}
The classes $T(\CF)$ and $U(\etol)$ are incomparable.
\end{corollary}
\begin{proof}
This follows from Propositions~\ref{basics} and \ref{et0l}.
\qed
\end{proof}

\begin{proposition}\label{cflin}
$T(\CF)$ is a proper subclass of $U(\lin)$.
\end{proposition}

\begin{proof}
  Let $\rho$ be a relation on $X^*$ such that $\rho = L\pi$ for a context-free language $L$. Let
  $\Gamma = (N, \Xtt, P, S)$ be a context-free grammar in Chomsky normal form for $L$.  We construct a linear indexed grammar with non-terminals $\set{I, I_A\mid A\in N}$, start symbol $I$, flags $N\cup\set{\$}$ (written as superscripts on the non-terminals) and productions
\begin{align*}
I   & \ra I_S^{\$} &&                             & I_A    & \ra aI & \text{for $A\ra (a,\emptyword)$ in $P$} \\
I_A & \ra I_B^C    && \text{for $A\ra BC$ in $P$} & I_A    & \ra Ia & \text{for $A\ra (\emptyword,a)$ in $P$} \\
I^A & \ra I_A      && \text{for $A\in N$}         & I^{\$} & \ra \#.
\end{align*}
The idea is that the grammar uses the flags to perform a leftmost derivation
in $\Gamma$, with the symbols from the first tape being output to the left and
the symbols from the second tape to the right (so that they appear in reverse).
Each sentential form $(w_1,w_2)B_1\ldots B_k$ in a leftmost derivation in $\Gamma$
corresponds in $\Gamma'$ to the sentential form $w_1 I_{B_1}^{B_2\ldots B_k \$} w_2^\rev$.
When the derivation is complete, we will have obtained the expression
$u I^{\$} v^\rev$ for some $(u,v)\in \rho$.  Finally, the production $I^{\$}$
places the symbol $\#$ to give $u\#v^\rev$, as desired.
Hence $L_\rho$ is linear indexed and so $T(\CF)\subseteq U(\lin)$.
\qed
\end{proof}

The incomparability of $\lin$ and $\etol$ is a straightforward consequence of results in
\cite{duske_linear,ehrenfeucht_relationship,kallmeyer_parsing}; for details, see \cite[Proofs of Corollaries 4.2
and~4.3]{bcm_crosssection}. Hence Proposition~\ref{basics} implies that amongst the classes $U(\lin)$, $U(\etol)$,
$T(\lin)$ and $T(\etol)$, the only containments are $U(\lin) \subseteq T(\lin)$ and $U(\etol) \subseteq T(\etol)$.




\begin{conjecture}\label{1cet0lconj}
$T(\OC)$ and $U(\etol)$ are incomparable.
\end{conjecture}

Let $\rho$ be the relation defined in Example~\ref{samelensamesymbolcount}; recall that $\rho$ was proven to lie in
$T(\OC)$. It appears unlikely that $L_\rho$ is ET0L, since there seems to be no way of using ET0L tables to maintain both the same length of words on each side of $\#$ and the same number of symbols $a$, while allowing the occurrences of $a$ to appear in any combination of
positions.

\begin{conjecture}\label{linconj}
$T(\lin)$ and $U(\ind)$ are incomparable.
\end{conjecture}

The relation $\sigma$ mentioned as a potential witness for Conjecture~\ref{indconj}
would also serve as a witness that $T(\lin)\not\subseteq U(\ind)$.

\begin{proposition}
$T(\ind)$ is properly contained in $U(\csl) = T(\csl)$.
\end{proposition}
\begin{proof}
We have $T(\ind)\subseteq T(\csl)$ by containment of language classes,
while $U(\ind)\setminus T(\ind)$ is non-empty by Proposition~\ref{basics},
and by Proposition~\ref{csl}, $U(\csl) = T(\csl)$.
\qed
\end{proof}

\section{The one-symbol case}


Let $\mathcal{B}$ consist of all binary relations on sets of words over one-symbol alphabets. Let
\[
  T_1(\C) = T(\C) \cap \mathcal{B}\text{ and }U_1(\C) = U(\C) \cap \mathcal{B}.
\]

\begin{proposition}\label{1ratlin}
  $T_1(\rat) = T_1(\lin)$.
\end{proposition}

\begin{proof}
  Let $\rho \in T_1(\lin)$ and let $\Gamma$ be a two-tape context-free grammar for $\rho$. Without loss of generality,
  assume that the set of terminal symbols of $\Gamma$ is
  $A = \set{(a,\emptyword),(\emptyword,a)}$.
  The Parikh image of
  the language over $A$ generated by $\Gamma$ is a semi-linear set $S \subset \N_0 \times \N_0$
 \cite[Theorem~5.1]{duske_linear}.
  Therefore
  $\rho = \gset{(a^\alpha,a^\beta)}{(\alpha,\beta) \in S}$, and so $\rho$ is a rational relation and hence
  $\rho \in T_1(\rat)$.
  \qed
\end{proof}

\begin{proposition}\label{1rat1c}
  $T_1(\rat) = U_1(\OC)$
\end{proposition}

\begin{proof}
  Let $\rho \in T_1(\rat)$ and consider a transducer $T$ that recognizes $\rho$. Construct a one-counter automaton $C$
  that accepts a word $u\#v^\rev$ if and only if $(u,v)$ is accepted by $T$ as follows. The one-counter automaton $C$
  keeps in its state a simulated copy of the state of $T$, beginning with its start state. At some point before it
  reaches $\#$, while its simulated state is $q$, it nondeterministically selects a transition of $T$ starting at
  $q$. Suppose this transition has label $(a^k,a^\ell)$ (for $k,\ell \in \N \cup \set{0}$) and leads to state $p$. Then
  $C$ reads $k$ symbols $a$ from its input, failing if it reads $\#$, and increments its counter by $\ell$. When the
  simulated state is a final state of $T$, the automaton $C$ can read $\#$. After having read $\#$, the automaton $C$
  simply reads $v$ symbol-by-symbol, decrementing its counter by $1$ each time, accepting if the counter is $0$ when the
  end of the input is reached. It is clear that $u\#v^\rev$ is accepted if and only if the numbers of symbols $a$ making
  up $u$ and $v$ are the numbers of symbols $a$ on the two sides of a path leading from a start to a final state of $T$. \qed
\end{proof}

Thus we have established that all classes of relations inside the box in Figure~\ref{hierarchy}
are equal when $|X|=1$.

\begin{corollary}
If $\mathfrak{C}$ is any class of languages intermediate between
$\rat$ and $\lin$ inclusive, then $T_1(\mathfrak{C}) = T_1(\rat)$.
If $\mathfrak{D}$ is any class of languages intermediate between
$\OC$ and $\lin$ inclusive, then $U_1(\mathfrak{D}) = T_1(\rat)$.
\end{corollary}
\begin{proof}
The first statement follows immediately from Proposition~\ref{1ratlin},
while the second follows from Propositions~\ref{1ratlin}, \ref{1rat1c}
and the $\lin$ case of Proposition~\ref{indexed}.
\qed
\end{proof}

However, this is as far as the equality extends (amongst language classes considered in this paper). On the one hand, the identity function $e: \N\ra \N$ in Example~\ref{functions} shows that $U_1(\rat)$ is properly contained in $T_1(\rat)$. On the other hand, we have the
following:

\begin{proposition}\label{1ratet0l}
$T_1(\rat)$ is properly contained in $U_1(\etol)$.
\end{proposition}
\begin{proof}
By Proposition~\ref{ratcf}, $T_1(\rat)\subseteq U_1(\CF)$; hence $T_1(\rat)\subseteq U_1(\etol)$.

Let $K = \gset{x^{2^n}}{ n\in \N}$ and $\rho_K = \gset{(w,w)}{w\in K}$.
Then $L_{\rho_K}$ is generated by
an ET0L-system with axiom $S\#S$ and two tables consisting of the single
productions $S\ra SS$ and $S\ra x$.  But $\rho$ does not lie in $T_1(\CF)$
since the projections onto each tape would then also be context-free and
$K$ is not context-free. Hence $U_1(\etol)\neq T_1(\rat)$.
\qed
\end{proof}

It remains open whether there are any equalities between
$U_1(\etol)$, $T_1(\etol)$, $U_1(\ind)$ and $T_1(\ind)$.

\section{Word problems of monoids}
\label{sec:wordprob}

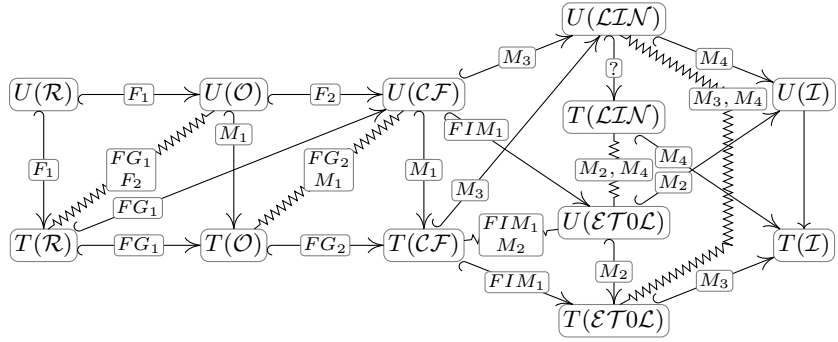
\begin{figure}[t]
  \centering
  \begin{tikzpicture}[x=25mm,y=20mm]
    \useasboundingbox (-.5,.8) -- (6,-1.8);
    \begin{scope}[
      every node/.style={
        rectangle with rounded corners,
        draw=gray,
        font=\small,
        inner sep=.5mm,
      }
      ]
      \node (urat) at (0,0) {$U(\rat)$};
      \node (trat) at (0,-1) {$T(\rat)$};
      \node (uone) at (1,0) {$U(\OC)$};
      \node (tone) at (1,-1) {$T(\OC)$};
      \node (ucf) at (2,0) {$U(\CF)$};
      \node (tcf) at (2,-1) {$T(\CF)$};
      \node (ulin) at (3,.5) {$U(\lin)$};
      \node (tlin) at (3,-.15) {$T(\lin)$};
      \node (uetol) at (3,-.85) {$U(\etol)$};
      \node (tetol) at (3,-1.5) {$T(\etol)$};
      \node (uind) at (4,0) {$U(\ind)$};
      \node (tind) at (4,-1) {$T(\ind)$};
    \end{scope}
    \begin{scope}[
      every node/.style={
        font=\scriptsize,
        fill=white,
        draw=gray,
        rectangle with rounded corners,
        rectangle corner radius=.5mm,
        inner sep=.3mm,
        align=center,
      },
      ]
      \draw[contained] (urat) edge node {$F_1$} (uone);
      \draw[contained] (uone) edge node {$F_2$} (ucf);
      \draw[rcontained] (ucf) edge node {$M_3$} (ulin);
      \draw[contained] (ucf) edge node[pos=.2] {$FIM_1$} (uetol);
      \draw[contained] (ulin) edge node {$M_4$} (uind);
      \draw[rcontained] (uetol) edge node[pos=.25] {$M_2$} (uind);
      \draw[contained] (trat) edge node {$FG_1$} (tone);
      \draw[contained] (tone) edge node {$FG_2$} (tcf);
      \draw[contained] (tcf) edge node {$FIM_1$} (tetol);
      \draw[contained] (tlin) edge node[pos=.25] {$M_4$} (tind);
      \draw[rcontained] (tetol) edge node {$M_3$} (tind);
      \draw[contained] (urat) edge node {$F_1$} (trat);
      \draw[contained] (uone) edge node[pos=.2] {$M_1$} (tone);
      \draw[contained] (ucf) edge node {$M_1$} (tcf);
      \draw[contained] (ulin) edge node {?} (tlin);
      \draw[contained] (uetol) edge node {$M_2$} (tetol);
      \draw[maybecontained] (uind) edge (tind);
      \draw[incomparable] (trat.75) -- node {$FG_1$\\$F_2$} (uone.220);
      \draw[incomparable] (tone) -- node {$FG_2$\\$M_1$} (ucf);
      \draw[incomparable] (tcf) -- node {$FIM_1$\\$M_2$} (uetol);
      \draw[incomparable] (tlin) -- node {$M_2,M_4$} (uetol);
      \draw[incomparable] (ulin.290) -- ($ (uind) + (-.4,0) $) -- node[pos=.03] {$M_3,M_4$} ($ (tind) + (-.4,0) $) -- (tetol.70);
      \draw[rcontained] (trat) edge node[pos=.2] {$FG_1$} (ucf);
      \draw[rcontained] (tcf) edge node[pos=.2] {$M_3$} (ulin);
    \end{scope}
  \end{tikzpicture}
  \caption{Illustration of separation of classes using word problems; for the notation used for monoids, see
    \S~\ref{sec:wordprob}.}
  \label{wordproblemsep}
\end{figure}

For a monoid $M$ with finite generating set $X$,
the \emph{word problem relation} of $M$ with respect to $X$ is
$\iota(M,X) = \gset{(u,v)}{u,v\in X^*, u=_M v}$.
Note that this is by definition an equivalence relation.
We say that $M$ has word problem in $T(\C)$ if
$\iota(M,X)\in T(\C)$, and that $M$ has word problem in
$U(\C)$ if $L_{\iota(M,X)}\in U(\C)$.
Note that $L_{\iota(M,X)}$ was the first language-theoretic version of
monoid word problems to be studied, and is often denoted $\WP(M,X)$.

In this section we exhibit examples of monoid word problems distinguishing the relation classes under consideration.
Besides demonstrating the algebraic relevance of the relation classes, these examples also establish that the hierarchy
shown in Figure~\ref{hierarchy} still holds when we restrict our attention to equivalence relations. The separation of
the various classes by word problems is summarized, using the notation of this section, in
Figure~\ref{wordproblemsep}. For reasons of space, proofs of more technical results are omitted from this section and given in an appendix.

Denote the free monoid, the free inverse monoid, and the free group of rank $n$ by, respectively, $FM_n$, $FIM_n$, and
$FG_n$.

\begin{proposition}\label{wprat}
  The word problem of the free monoid of rank~1 is in $T_1(\rat)$ and in $U_1(\OC)$ but not in $U_1(\rat)$.
\end{proposition}

\begin{proof}
Let $F_1$ be generated by $x$ and let $\rho = \iota(F_1,\set{x})$.  We have
$\rho = \gset{ (x^n, x^n)}{n\in \N_0}$, which is the relation of the 
identity function $e(n) = n$.  We showed in Example~\ref{functions} that this
relation is in $T(\rat)\setminus U(\rat)$.
Moreover, $L_\rho = \gset{x^n\# x^n}{n\in \N_0}$ is one-counter.
\qed
\end{proof}

\begin{proposition}\label{wpratcf}
The free group of rank~1 has word problem in $U(\OC)$
but not in $T(\rat)$. The free monoid of rank greater than~1 has word problem in $T(\rat)$ but not in $U(\OC)$.
\end{proposition}

\begin{proof}
Let $F$ be the free group of rank~$1$, generated
as a monoid by $X = \{x,x^{-1}\}$.
Then $L_{\iota(F,X)} = \gset{u\#v }{ u,v\in X^*,
|u|_x - |u|_{x^{-1}} = |v|_x - |v|_{x^{-1}}}$.
This language is easily recognised by a one-counter
automaton that uses the stack to calculate
$|u|_x - |u|_{x^{-1}}$ and then checks this against
$|v|_x - |v|_{x^{-1}}$, accepting by empty stack.
We require four states to record whether we are
currently reading $u$ or $v$ and whether we currently
have an excess of $x$'s or of $x^{-1}$'s.
The only groups $G$ with $\iota(G)$ rational are
finite \cite[Theorem~8.7.9]{pfeiffer_phd}, so $\iota(F,X)$ is not rational.

Let $|X|\geq 2$ and let $\rho$ be the equality relation on $X^*$, which is the two-tape word problem of $X^*$.
Then $\rho$ is rational, but $L_\rho = \gset{w\# w^\rev}{ w\in X^*}$
is not one-counter by \cite[Proposition~4.1]{holt_onecounter} since $X^*$ has exponential growth.
\qed
\end{proof}

\begin{proposition}\label{wp1ccf1}
The word problem of the free group of rank~2 is in $U(\CF)$ but not in $T(\OC)$.
\end{proposition}

\begin{proof}
Let $F_2$ be the free group on $X=\set{x,y}$ and let $X^\pm = \{x,y,x^{-1},y^{-1}\}$.
The word problem of $F_2$ is well known to be context-free:
$L_{\iota(F_2,X^\pm)}$ is accepted by a pushdown automaton that
pushes each symbol read onto the stack, unless it is the inverse of the current
top-of-stack symbol (in which case the stack is popped) or $\#$ (in which case the stack
is unchanged).
Let $W = \iota(F_2,X)$ and suppose $W$ is one-counter.
Then the set $W_1$ of all pairs $(w,\epsilon)$ in $W$ is also one-counter
(we can modify an automaton accepting $W$ to move to a failure state
if any symbols are read from the second tape).
But $W_1$ is equivalent to the group word problem of $F_2$
(the set of all words equal to the identity in $F_2$),
and a group has one-counter word problem if and only if it is cyclic \cite{holt_onecounter}.
\qed
\end{proof}

\begin{proposition}{\rm \cite[Theorem~1 and Corollary~1]{brough_inverse}}
The free inverse monoid of rank $1$
has word problem in $U(\etol)$ but not in $T(\CF)$.
\end{proposition}

\begin{proposition}\label{wp1ccf2}
Let $X = \set{a,b,\ell,r}$ and let $M_1$ be the monoid with presentation
\[ \langle X \mid \ell a^n b^n r =  \ell b^n a^n r \; (n\in \N)\rangle. \]
Then the word problem of $M_1$ is in $T(\OC)$ but not in $U(\CF)$.
\end{proposition}

\begin{proof}
Let $\rho = \iota(M_1,X)$.  Define a one-counter automaton $\A$ with
initial and final state $q_0$ and further states
$p_1,q_1,p_2,q_2$ as follows:
In state $q_0$, $\A$ reads $(x,x)$ for $x\in X$, or on input $(\ell,\ell)$,
$\A$ may move to state $p_1$ or $q_1$.  In state $p_1$, $\A$ reads $(a,b)$ and increments the counter, or moves to state $q_2$ on input $(b,a)$, decrementing the counter.
In state $p_2$, $\A$ reads $(b,a)$ and decrements the counter, or moves to
state $q_0$ on input $(r,r)$ if and only if the counter is at $0$.
States $q_1$, $q_2$ act the same as $p_1,p_2$ respectively, except with the roles
of $a$ and $b$ swapped.
Then $\A$ accepts exactly all pairs $(u,v)$ such that $u$ and $v$ are
equal in all positions, except that subwords of the form $\ell a^n b^n r$ and
$\ell b^n a^n r$ in corresponding positions may be interchanged.  That is,
$\A$ accepts $\rho$, so $\rho\in T(\OC)$.

Let $\phi$ be the homomorphism on $X\cup \set{\#}$ that maps $\ell,r$ and
$\#$ to $\emptyword$.
Then
\[ L_1 = \phi( L_\rho\cap \ell a^* b^* r \# r a^* b^* \ell) = \gset{a^n b^n a^n b^n}{n\in \N_0}\]
is not context-free, hence $L_\rho$ is also not context-free, so $\rho\notin U(\CF)$.
\qed
\end{proof}

Our remaining examples are based on two general construction techniques
similar to that used in the previous proposition.
For $\rho\subseteq X^*\times X^*$, define a monoid
$M[\rho] = \langle X, \ell, r \mid \ell u r = \ell v r \; (u \mathrel\rho v)\rangle$.
For $L\subseteq X^*$, define a monoid
$M(L) = \langle X, \ell, r \mid \ell w r = \ell r \; (w\in L) \rangle$.

\begin{proposition}\label{wpconstructs}
Let $\C$ and $\D$ be classes of languages.
\begin{enumerate}
\item If $\rho\in T(\C)$ and $\C$ is closed under union, concatenation and Kleene star
and contains $\rat$,
then $M[\rho]$ has word problem in $T(\C)$.
\item If $\C\in \set{\CF,\etol,\edtol,\lin,\ind}$ and
$L\in \C\setminus \D$ and furthermore $\D$ is closed
under transductions,
then $M(L)$ has word problem in $U(\C)\setminus T(\D)$.
\end{enumerate}
\end{proposition}

\begin{proof}
First, let $\rho\subseteq X^*\times X^*$ be in $T(\C)$.
Let $K\subset (\Xtt)^*$ such that $\rho = K\pi$,
$Y = X\cup\set{\ell,r}$, $Y_2 = \gset{(y,y)}{y\in Y}$ and
$\sigma = \iota(M_\rho,Y)$.
Then $\sigma = K'\pi$ for $K' = Y_2^* \cup Y_2^* (\ell,\ell) K (r,r)^* Y_2^*$.
Since $Y_2^*$ is rational, by the closure properties of $\C$ we have $K'\in \C$
and hence $\sigma\in T(\C)$.

Second, let $L\subseteq X^*$ be in $\C\setminus \D$, with $Y$ and $Y_2$ as before,
and let $\sigma_1 = \iota(M(L),Y)$.
Suppose $\sigma_1\in T(\D)$ with $\sigma' = K_1\pi$ for $K_1\in \D$.
If $\D$ is closed under transductions, then the sublanguage $K'_1$ of $K_1$ consisting
of all $w$ such that $w\pi$ has $\ell r$ on the second tape is also in $\D$.
But then $K'_1\pi = \gset{(\ell r, \ell r), (\ell w r, \ell r)}{w\in L}$ is in $T(\D)$,
implying that $\gset{\ell r, \ell w r}{w\in L}$ (the projection of $K_1'\pi$ onto the
first tape) is in $\D$.  But applying another transduction to this
implies that $L$ itself is in $\D$, which is false.  Hence $\sigma_1 = \iota(M(L),Y)$ is not in $T(\D)$.
On the other hand,
$L_{\sigma_1}$ consists of all $u\# v^\rev$ such that either $u=v$ or
$u = \alpha_0 \ell u_1 r \alpha_1 \ldots \ell u_n r \alpha_n$
and $v = \alpha_0 \ell v_1 r \alpha_1 \ldots \ell v_n r \alpha_n$
for $\alpha_i\in Y^*$ and $(u_i,v_i)\in (L\times \{\emptyword\})\cup (\{\emptyword\}\times L)$.
Let $\Gamma$ be an indexed grammar for $L$.  Then an indexed grammar for
$L_{\sigma_1}$ of the same `type' as $\Gamma$ (context-free, ET0L-indexed etc.)
with start symbol $I$ can be defined by setting the following productions from $I$:
\[
I\ra aIa \mid \ell S r I r \ell \mid \ell r I r S' \ell \mid \# \quad (a\in Y),
\]
where $S$ is the start symbol of $\Gamma$ and $S'$ is the start symbol of
the grammar $\Gamma^\rev$ for $L^\rev$.  Hence $\iota(M(L),Y) = \sigma_1$
is in $U(\C)$.
\qed
\end{proof}

Let $X = \{x,y\}$ and let $K\subset X^*$ be a context-free language which is not ET0L, as in the proof of
Proposition~\ref{et0l}.  Let $X' = \{x' \mid x\in X\}$ and define a homomorphism $\phi:X^*\ra (X')^*$ by $x\phi = x'$.
Let $\rho_K = \gset{(w,w\phi)}{w \in K}$ and $M_2 = M[\rho_K]$.

\begin{proposition}\label{wpetol}
The monoid $M_2$ has word problem in $T(\CF)$
but not in $U(\etol)$.
\end{proposition}

\begin{proof}
Let $X = \{x,y\}$ and let
$K\subset X^*$ be a context-free language which is not ET0L,
as in the proof of Proposition~\ref{et0l}.
Let $X' = \{x' \mid x\in X\}$ and define a homomorphism
$\phi:X^*\ra (X')^*$ by $x\phi = x'$.
Then the language $\gset{w(w\phi)^\rev}{ w\in K}$ is
not ET0L \cite{ehrenfeucht_relationship}.
Let $Y=  \{r,\ell\}\cup X\cup X'$
and let $M_2$ be the monoid with presentation
\[ \langle Y \mid \ell w r = \ell(w\phi)r \; (w\in K)\rangle. \]
Then $L_{\iota(M_2,Y)}\cap \ell X^*r \# r(X')^*\ell =
\gset{\ell w r \# r(w\phi)^\rev\ell }{ w\in K }$.
This language has as a homomorphic image the non-ET0L language
mentioned above, and since the ET0L languages are closed under
homomorphism, $L_{\iota(M_2,Y)}$ cannot be ET0L.

However, $M_2 = M[\rho]$ for the relation $\rho = \gset{(w,w\phi)}{w\in K}$,
which is in $T(\CF)$.  Hence $M_2$ has word problem in $T(\CF)$ by
Proposition~\ref{wpconstructs}, as $\CF$ satisfies the required closure properties.
\qed
\end{proof}

\begin{proposition}
There exists a monoid $M_3$ with word problem in
$U(\lin)$ but not in $T(\etol)$, and a monoid $M_4$
with word problem in $U(\etol)$ but not in $T(\lin)$.
\end{proposition}

\begin{proof}
This follows immediately from Proposition~\ref{wpconstructs} and the
incomparability of $\etol$ and $\lin$.
\qed
\end{proof}

Let $\rho = \gset[\big]{(a_1^n a_2^n a_3^n,b_1^n b_2^n b_3^n)}{n \in \N}$ and let $M_5 = M[\rho]$. (Recall that the
relation $\rho$ was used in the proof of Proposition~\ref{linind}.)

\begin{proposition}
  \label{m5}
The monoid $M_5$ has word problem in $T(\lin)$ but not in $U(\lin)$.
\end{proposition}

\begin{proof}
Let $X = \{a_1,a_2,a_3,b_1,b_2,b_3,\ell,r\}$ and
\[
M_5 = \langle X \mid \ell a_1^n a_2^n a_3^n r  = \ell b_1^n b_2^n b_3^n r \; (n\in \N)\rangle.
\]
Then $M_5 = M[\rho]$ for the relation $\rho$ given in the proof of
Proposition~\ref{wpconstructs}, which is in $T(\lin)$ but not in $U(\lin)$.
Hence by Proposition~\ref{wpconstructs}, $\iota(M_5,X)\in T(\lin)$.
Intersecting $L_{\iota(M_5,X)}$ with $\ell a_1^* a_2^* a_3^* r \# r b_3^* b_2^* b_1^* \ell$
and applying a homomorphism yields $L_\rho$, hence $\iota(M_5,X)\notin U(\lin)$.
\qed
\end{proof}

\section{Relations of larger arity}

The two-tape definition of binary relations is straightforwardly extendable to relations
of any arity.  It is less clear how to define the `unfolded' version for non-binary relations.
Indeed, which definition makes sense is likely to depend on the semantic content of
the relation in question.
For example, Gilman \cite{gilman_wordhyperbolic} defines the multiplication table of a group
with respect to a regular language $L$ of normal forms
as $\gset{ u\# v\# w^\rev}{ u,v,w\in L, uv =_G w}$ (and the same definition is extended to semigroups
in \cite{duncan_hyperbolic}).  However, for other ternary relations, such as the various betweenness
relations, it might make little or no sense to reverse any of the component words.
Meanwhile, for quaternary relations, in some cases the unfolding
$(u,v,w,x)\mapsto u\# v\# w\# x^\rev$ will make sense, as in the case where the relation
represents some ternary operation with $u\cdot v\cdot w = x$; while in others a more
natural unfolding could be $(u,v,w,x)\mapsto u\#v\#x^\rev\# w^\rev$
(not even preserving the order of the components), for example
if the relation consisted of tuples such that $uv = wx$ in some algebraic structure.

\bibliography{\jobname}
\bibliographystyle{elsarticle-num}

\end{document}